\pdfoutput=1
\RequirePackage{amsmath}
\documentclass[runningheads]{llncs}
\usepackage[T1]{fontenc}
\usepackage{graphicx}
\usepackage{hyperref}
\usepackage{color}

\usepackage{amssymb}
\usepackage{mathtools}
\usepackage{multirow}
\usepackage{wrapfig}

\usepackage[ruled,linesnumbered]{algorithm2e}
\usepackage[caption=false,font=footnotesize]{subfig}

\newcommand{\len}[1]{|#1|}

\newcommand{\seq}[1]{\langle #1 \rangle}

\newcommand{\sk}[1]{\tilde{#1}}
\newcommand{\ske}{\sk{e}_i}

\newcommand{\skl}{\sk{L}}
\newcommand{\set}[1]{\{#1\}}
\newcommand{\mi}[1]{\mathit{#1}}

\begin{document}
\title{Ranking the Top-K Realizations of Stochastically Known Event Logs}
%
%
\author{Arvid Lepsien\inst{1}\orcidID{0000-0002-8105-382X} \and
Marco Pegoraro\inst{2}\orcidID{0000-0002-8997-7517} \and
Frederik Fonger\inst{1}\orcidID{0009-0000-8445-8104} \and
Dominic Langhammer\inst{3}\orcidID{0000-0002-0238-324X} \and
Milda Aleknonyt\.{e}-Resch\inst{1}\orcidID{0000-0003-0472-1262} \and
Agnes Koschmider\inst{3}\orcidID{0000-0001-8206-7636}}
\authorrunning{A. Lepsien et al.}
\institute{Department of Computer Science, Kiel University, Kiel, Germany\\
\email{\{ale, ffo, mar\}@informatik.uni-kiel.de} \and
Chair of Process and Data Science (PADS), RWTH Aachen University, Aachen, Germany\\
\email{pegoraro@pads.rwth-aachen.de} \and
Chair of Business Informatics and Process Analytics, University of Bayreuth, Bayreuth, Germany\\
\email{firstname.lastname@uni-bayreuth.de}}
\maketitle              
\begin{abstract}
Various kinds of uncertainty can occur in event logs, e.g., due to flawed recording, data quality issues, or the use of probabilistic models for activity recognition.
Stochastically known event logs make these uncertainties transparent by encoding multiple possible realizations for events.
However, the number of realizations encoded by a stochastically known log grows exponentially with its size, making exhaustive exploration infeasible even for moderately sized event logs.
Thus, considering only the top-K most probable realizations has been proposed in the literature.
In this paper, we implement an efficient algorithm to calculate a top-K realization ranking of an event log under event independence within O(Kn), where n is the number of uncertain events in the log.
This algorithm is used to investigate the benefit of top-K rankings over top-1 interpretations of stochastically known event logs.
Specifically, we analyze the usefulness of top-K rankings against different properties of the input data. 
We show that the benefit of a top-K ranking depends on the length of the input event log and the distribution of the event probabilities.
The results highlight the potential of top-K rankings to enhance uncertainty-aware process mining techniques.

\keywords{Event Data \and Uncertainty \and Top-K Ranking \and Algorithm}
\end{abstract}
\section{Introduction}
\label{introduction}

Process mining is a family of techniques for analyzing event logs, providing data-driven insights and enabling objectively informed decision making \cite{van_der_aalst_process_2012}.
However, a notable challenge arises when event logs contain events affected by uncertainty~\cite{koschmider_process_2024}.
This uncertainty can emerge for various reasons, such as through systems applying probabilistic models for activity recognition~\cite{engelberg_uncertainty-aware_2023}\cite{lepsien_analytics_2023}. 
To manage this uncertainty, stochastically known event logs have been introduced~\cite{gal_everything_2023}, which encode multiple possible realizations for each of their events.
Despite this, most current process mining techniques do not account for this uncertainty, and instead rely only on the most probable (top-1) interpretation of the events.

To illustrate this, let us consider an example process of patient treatment in a hospital where activities are extracted from medical note systems and handwritten notes. 
Techniques such as handwriting recognition and natural language processing are used to detect these activities~(e.g.,~\cite{kecht_event_2021}).
Table~\ref{tbl:skl} shows an example of a stochastically known event log from this setting. 
In this example, the physician begins the treatment by establishing the medical history of the patient (activity \textit{H}). 
The log then shows uncertainty regarding the diagnosis of either light pain (activity \textit{L}) or severe pain (activity \textit{S}).
Similarly, there is uncertainty in the subsequent event, where either ibuprofen (activity \textit{I}) or opiates (activity \textit{O}) were prescribed. Table~\ref{tbl:realizations} lists all possible realizations of this log with their realization probabilities.
While the most probable realization $\langle H, L, I \rangle$ complies with the hospital's guidelines, the second and third most probable realizations $\langle H, L, O \rangle$ and $\langle H, S, I \rangle$ -- which in sum are more probable than the first realization -- both hint at compliance issues.
If only the most probable log realization was considered, these issues would be overlooked.

While the complete set of possible realizations can be calculated and analyzed efficiently for this log with few uncertain events, the computational effort grows exponentially as the size of the event log increases.
This raises the need for efficient techniques to prioritize the realizations of uncertain logs.
One initial approach suggests considering only the top-$K$ most probable log realizations~\cite{gal_everything_2023}. 
However, neither an efficient algorithm to calculate the top-$K$ realizations of a stochastically known log nor an evaluation demonstrating the feasibility of top-$K$ rankings for uncertainty-aware process mining techniques have been presented yet.
These issues are addressed in this paper, guided by the following research questions (RQs):
\begin{itemize}
    \item \textbf{(RQ1)} How to efficiently calcuate top-$K$ rankings?
    \item \textbf{(RQ2)} What is the benefit of top-$K$ rankings over top-$1$ interpretations in terms of covered probability mass?
    \item \textbf{(RQ3)} How well can top-$K$ rankings represent the variability of the possible log realizations encoded by a stochastically known log?
\end{itemize}
To answer these RQs, we design a basic top-$K$ algorithm for stochastically known logs.
This algorithm is based on a generalized procedure that iteratively partitions the set of possible log realizations~\cite{hamacher_k_1985}.
We then apply the algorithm to simulated stochastically known logs to evaluate the potential of top-$K$ rankings to improve the understanding of uncertainty-aware event data.
The remainder of the paper is structured as follows.
Sec.~\ref{preliminaries} introduces the basic notations.
Sec.~\ref{related_work} discusses related literature.
The algorithm to produce the top-$K$ realizations of a stochastically known event log is presented in Sec.~\ref{algorithm}.
Then, the efficiency of this algorithm and the general potential of top-$K$ rankings to improve the understanding of stochastically known event logs are evaluated in Sec.~\ref{sec:evaluation}.
The paper concludes with a summary and an outlook in Sec.~\ref{conclusion}.

\section{Preliminaries}
\label{preliminaries}

In this section, the definitions and notations used in the paper (mostly based on~\cite{pegoraro_conformance_2021}) are summarized. 
First, the universes are defined, which are then used to formally define event logs.

\begin{definition}[Universes~\cite{pegoraro_conformance_2021}]
    Let $\mathcal{U}_I$ be the universe of event identifiers.
    Let $\mathcal{U}_C$ be the universe of case identifiers.
    Let $\mathcal{U}_A$ be the universe of activities and let $\mathcal{U}_T$ be the totally ordered set of timestamp identifiers.
\end{definition}

\begin{definition}[Event, event log \cite{pegoraro_conformance_2021}]
    We denote with $\mathcal{E}_C = \mathcal{U}_I \times \mathcal{U}_C \times \mathcal{U}_T \times \mathcal{U}_A$ the universe of deterministic events.
    A deterministic event log is a set of events $L_C \subseteq \mathcal{E}_C$ such that every event identifier in $L_C$ is unique.
\end{definition}
Next, the notion of an event is extended to encode uncertainty by replacing the single deterministic activity with a partial function matching multiple alternative activities and their corresponding confidence.
\begin{definition}[Stochastically known event, stochastically known event log \cite{pegoraro_conformance_2021}]\label{def:stoch_known_event}
    We denote with $\mathcal{E}_W = \set{(e_i, c, t, f) \in \mathcal{U}_I \times \mathcal{U}_C \times \mathcal{U}_T \times (\mathcal{U}_A \not\rightarrow [0, 1]) \mid \sum_{a \in \mi{dom}(f)} f(a) = 1}$ the universe of stochastically known events.
    A stochastically known (event) log is a set of stochastically known events $\skl \subseteq \mathcal{E}_W$ such that every event identifier in $\skl$ is unique.
\end{definition} 

An example case of a stochastically known log is shown in Table~\ref{tbl:skl}.
We use a tilde to distinguish between stochastically known events and event logs and their deterministic counterparts (e.g., $\ske \in \skl$ vs. $e_i \in L$). 
\begin{table}[!t]
\renewcommand{\arraystretch}{1.3}
\begin{minipage}{.45\columnwidth}
\caption{Exemplary stochastically known event log}
\label{tbl:skl}
\centering
\begin{tabular}{cccc}
\hline
Event    & Case     & t         & Activity \\ \hline
1        & 1        & 1         & H         \\ 
2        & 1        & 2         & \{(L, 0.7), (S, 0.3)\}                 \\ 
3        & 1        & 3         & \{(I, 0.6), (O, 0.4)\}         \\
\hline
\end{tabular}
\end{minipage}
\hfill
\begin{minipage}{.45\columnwidth}
\renewcommand{\arraystretch}{1.3}
\caption{Possible realizations of the log shown in table~\ref{tbl:skl}}
\label{tbl:realizations}
\centering
\begin{tabular}{cc}
\hline
$L \in R(\skl)$                     & $P_{\skl}(L)$ \\ \hline
$\seq{H, L, I}$         & 0.42 \\ 
$\seq{H, L, O}$         & 0.28 \\
$\seq{H, S, I}$         & 0.18 \\ 
$\seq{H, S, O}$         & 0.12 \\ 
\hline
\end{tabular}
\end{minipage}
\end{table}
Next, the realizations of stochastically known events and event logs are defined.

\begin{definition}[Realizations]
\label{def:realizations}
    For a stochastically known event $\ske = (i, c, t, f)$, the set of realizations is defined as $R(\ske) = \set{(i, c, t, a) \mid a \in \mi{dom}(f)}$.
    For a stochastically known log $\skl \subseteq \mathcal{E}_W$, the set of realizations is defined as $R(\skl) = \bigtimes_{\ske \in \skl} R(\ske)$.
\end{definition}

The realizations of the exemplary stochastically known event log are shown in Table~\ref{tbl:realizations}.
In the following, a tilde over a deterministic event is used to refer to the stochastically known event it stems from, e.g., $e_i \in R(\ske)$ with $\ske \in \skl$.
Finally, the probability of event and log realizations is defined.

\begin{definition}[Realization probability]
    Let $\ske = (i, c, t, f)$ be a stochastically known event and $e_i \in R(\ske)$ be a realization of $\ske$.
    Let $\skl$ be a stochastically known log and $L \in R(\skl)$ be a realization of $\skl$.
    Then, the probability of the event realization $e_i$ is defined as $P_{\ske}(e_i) = f(e_i)$ and the probability of the log realization is defined as $P_{\skl}(L) = \prod\limits_{e_i \in L}P_{\ske}(e_i)$.
\end{definition}

When context is clear, the subscript of the probability function may be omitted.
In this paper, the realization probability is calculated under the assumption that the probability of any event realization does not depend on the realizations selected for, e.g., preceeding events (\textit{event independence}~\cite{gal_everything_2023}).
In practical applications, a weaker assumption of \textit{trace independence} -- which assumes independent probabilities across cases, but allows for dependent probabilities inside cases -- might be essential and we plan to consider it in future research~\cite{gal_everything_2023}.
\section{Related Work}
\label{related_work}

The interest in uncertainty-aware process mining has increased in recent years.
This is reflected by taxonomies classifying uncertain event data~\cite{pegoraro_conformance_2021,gal_everything_2023} and 
various approaches incorporating uncertainty into process mining~\cite{pegoraro_probabilistic_2022,pegoraro_conformance_2021,gal_everything_2023,DBLP:journals/eaai/FelliGMRW23}.
These approaches either analyze the information of all realizations of uncertain logs, e.g., with task-specific aggregations such as lower and upper bounds for directly-follows relation frequencies~\cite{pegoraro_probabilistic_2022}, or they select a single representative realization and thus recover a deterministic trace~\cite{bogdanov_sktr_2023}.

Stochastically known event logs have been considered in analogy to probabilistic databases~\cite{gal_everything_2023}.
There, uncertainty can be managed by considering the $K$ most important answers or possible interpretations of the uncertain data~\cite{gal_everything_2023,soliman_probabilistic_2008}.
For instance, Soliman et al.~\cite{soliman_probabilistic_2008} present an efficient algorithm for top-$K$ queries on probabilistic databases that allows the consideration of \textit{generation rules} in the form of mutually exclusive tuples.
Besides, top-$K$ rankings have also proven to be useful in additional applications, e.g., for search engine results~\cite{qin_diversifying_2012}, or combinatorial optimization~\cite{pascoal_note_2003,hamacher_k_1985}.
The application of top-$K$ rankings for the analysis of uncertain event data has been proposed in the literature~\cite{gal_everything_2023}.
To the best of our knowledge, no top-$K$ algorithm specifically designed for stochastically known logs has been presented yet.
The work of Pegoraro et al.~\cite{pegoraro_probability_2022} could be implicitly considered as a computation for the top-$K$ realizations of a trace, but in essence it is an inefficient brute force method.
Gal proposed an efficient top-$K$ ranking algorithm for stochastically known logs relying on~\cite{hamacher_k_1985}, where a generalized top-$K$ ranking procedure for combinatorial optimization problems is presented.
Bogdanov et al.~\cite{bogdanov_sktr_2023} present SKTR, a technique that retrieves the best realization of a stochastically known trace using a specially constructed graph where the paths correspond to the realizations of the stochastically known trace and the edge weights are determined based on a cost function that considers (dependent) event realization probabilities and conformance to a reference process model.
The application of a $K$ shortest paths algorithm (e.g.,~\cite{gao_fast_2010}) on the graph could result in top-$K$ ranking on the trace level.
More generally, a wide range of top-K algorithms for various combinatorial problems are available and could be adapted for uncertain event data by mapping the event data to fit the structure of such a problem, e.g., as has been suggested for assignment ranking algorithms~\cite{gal_everything_2023}.

\section{Top-K Realization Ranking}
\label{algorithm}

To investigate the potential of top-$K$ rankings for uncertainty-aware process mining, we first design an algorithm to efficiently calculate the top-$K$ realization of a stochastically known event log.
An intuitive approach is to first calculate all possible realizations, and then sort the results by probability to retrieve the top-$K$ realizations.
However, since the number of possible realizations grows exponentially with the number of events in the log, this approach is affected by exponential complexity.
Another approach would be adapting one of the approaches summarized in Sec.~\ref{related_work}.
These approaches impose constraints in terms of dependent probabilities, generation rules or reference process models.
These constraints generally make the resulting top-$K$ rankings more accurate because they exclude realizations that are irrelevant or meaningless in the analysis context, and thereby concentrate the probability mass among fewer realizations -- which is highly beneficial for using the resulting rankings.
However, in this paper, we employ a basic algorithm operating under the assumption of event independence.
This is done to explore the general utility of top-$K$ rankings for uncertain event data.
If the baseline top-$K$ rankings produced under event independence show a clear benefit over top-1 interpretations, it indicates that more complex approaches which produce more refined top-$K$ rankings will likely offer further advantages.
Our goal is to demonstrate that even under the independence assumption, top-$K$ rankings outperform top-1 interpretations, laying the groundwork for the potential benefits of more advanced techniques.

In the following, we present an efficient algorithm which is based on a generalized ranking procedure suggested in the literature~\cite{hamacher_k_1985}, addressing \textbf{RQ1}.
This algorithm is then used to evaluate the benefit of the produced top-$K$ rankings to address \textbf{RQ2} and \textbf{RQ3}.

\subsection{Top-K ranking algorithm}

First, the generalized terms for top-$K$ problems from~\cite{hamacher_k_1985} are introduced.
Hamacher and Queyranne~\cite{hamacher_k_1985} define a top-$K$ ranking as follows:

\begin{definition}[Top-$K$ ranking]
\label{def:topK}
    Let $\mathcal{D}$ be a set.
    Let $K \in \set{1, \ldots, \len{\mathcal{D}}}$.
    Let $c: \mathcal{D} \rightarrow \mathbb{R}$.
    Then, $D_1, \ldots, D_K$ is a top-$K$ ranking iff $\forall_{D \in \mathcal{D} \setminus \set{D_1, \ldots, D_K}}\ c(D_1) \leq \ldots \leq c(D_K) \leq c(D)$.
\end{definition}

\noindent Based on this, Hamacher and Queyranne define the general structure of a top-$K$ problem as a set of feasible solutions:

\begin{definition}[Feasible solutions \cite{hamacher_k_1985}]
    Let $E$ be a finite set and let $\mathcal{D} \subseteq 2^E$ be a set of subsets of $E$.
    We refer to the elements $e \in E$ as choices, and to the elements $D \in \mathcal{D}$ as (feasible) solutions.
    For any $I, O \subseteq E$ with $I \cap O = \varnothing$, $\mathcal{D}_{I, O} = \set{D \in \mathcal{D} \mid I \subseteq D \land D \cap O = \varnothing}$ is the set of feasible solutions restricted by $I$ and $O$.
\end{definition}

Hamacher and Queyranne present a generalized algorithm, called the BST procedure, that iteratively partitions the set of possible log realizations~\cite{hamacher_k_1985}.
This reduces the problem to a set of sub-problems, which can be solved using two auxiliary algorithms: (1) \texttt{ALG-1P} that yields the globally best solution and (2) \texttt{ALG-R2P} that yields the (locally) second best solution within a restricted solution set $\mathcal{D}_{I, O}$.
We apply this procedure to design an algorithm for the top-$K$ realizations problem for stochastically known logs, which is defined as follows.

\begin{definition}[Realization ranking problem]
\label{def:solutions}
    Let $\skl$ be a stochastically known log.
    Let $E = \bigcup_{\ske \in \skl} R(\ske)$ be the set of choices, and $\mathcal{D} = R(\sk{L})$ be the set of feasible solutions of the top-$K$ realizations problem.
    Then, $L_1, \ldots, L_K$ is a top-$K$ ranking iff $\forall_{L \in \mathcal{D} \setminus \set{L_1, \ldots, L_K}}\ P(L_1) \geq \ldots \geq P(L_K) \geq P(L)$, with $L_i$ being the realization having rank $i$.
\end{definition}

Because the goal of our algorithm is to \textit{maximize} the realization probability (instead of \textit{minimizing} a cost function as in~\cite{hamacher_k_1985}), the comparison operators in the ranking are flipped with respect to Def.~\ref{def:topK}.
Next, \texttt{ALG-1P} and \texttt{ALG-R2P} algorithms for the realization ranking problem are presented.
Under \textit{event independence}, the globally best log realization $L_1 \in R(\skl)$ is simply the log realization where each stochastically known event $\ske \in \skl$ is realized as its most probable alternative~\cite{bogdanov_sktr_2023}, so $\texttt{ALG-1P}(\skl) = \bigcup_{\ske \in \skl}\ \arg\max_{e_i \in R(\ske)}P(e_i)$.
For \texttt{ALG-R2P}, the following property of the realization ranking problem is shown:

\begin{lemma}
\label{lemma:second_best}
    For all $\mathcal{D}_{I, O} \subseteq \mathcal{D}$, if a second best solution $L_q^2 \in \mathcal{D}_{I, O}$ with $q \in \set{1, \ldots, K}$ exists, there always exists a second best solution $L' \in \mathcal{D}_{I, O}$ that is different from the best solution $L_q \in \mathcal{D}_{I, O}$ in exactly one element, i.e., $\len{L' \setminus L_q} = \len{L_q \setminus L'} = 1$.
\end{lemma}

\begin{proof}
    Let $L_q \in \mathcal{D}_{I, O}$ be a best solution in a restricted set of feasible solutions, and $L_q^2 \in \mathcal{D}_{I, O}$ be a second best solution with $\len{L_q^2 \setminus L_q} = \len{L_q \setminus L_q^2} > 1$.
    Without loss of generality, selecting any $e_i' \in L_q^2 \setminus L_q$ enables constructing a new solution $L' = (L_q \setminus \set{e_i}) \cup \set{e_i'}$, where clearly $\len{L' \setminus L_q} = \len{L_q \setminus L'} = 1$.
    Because it is constructed only out of choices from the valid solutions $L_q$ and $L_q^2$, $L'$ is also valid in $\mathcal{D}_{I, O}$.
    The constructed solution $L'$ substitutes one element in $L_q$, so $P(L')$ can be rewritten as $P(L') = P(L_q) \cdot \frac{P(e_i')}{P(e_i)}$.
    As $e_i \in L_q$ is a choice within the best solution, it is the most probable event realization, i.e., $P(e_i) \geq P(e_i')$, which implies $P(L_q) \geq P(L')$.
    By substituting the other elements out of $L_q^2 \setminus L_q$, the solution $(L_q \setminus (L_q \setminus L_q^2)) \cup (L_q^2 \setminus L_q) = L_q^2$ can be constructed step-wise from $L_q$.
    Through transitivity this implies $P(L_q) \geq P(L') \geq P(L_q^2)$.
\end{proof}

\begin{wrapfigure}[19]{r}{0.45\textwidth}
\vspace{2em}
\vspace{-\intextsep}
\begin{algorithm}[H]
\label{alg:r2p}
\SetCustomAlgoRuledWidth{0.45\textwidth}
\DontPrintSemicolon
\caption{\texttt{ALG-R2P}}
\KwData{$\skl, I, O, L_q$}
\KwResult{$L_q^2$}
$\rho_{\mi{max}} \gets 0$\;
\For{$e_i \in L_q \setminus I$}{
    $e_i' \gets \texttt{next}_{\skl}(e_i, \ske)$\;
    \If{$e_i' \in R(\ske) \setminus O$}{
        $\rho_i \gets \frac{P(e_i')}{P(e_i)}$\;
        \If{$\rho_i > \rho_{\mi{max}}$}{
            $\rho_{\mi{max}} \gets \rho_i$\;
            $e_{\mi{max}} \gets e_i$\;
            $e_{\mi{max}}' \gets e_i'$\;
        }
    }
    }

$L_q^2 \gets (L_q \setminus \set{e_{\mi{max}}}) \cup \set{e'_{\mi{max}}}$\;
\Return $L_q^2$\;
\end{algorithm}
\end{wrapfigure}
Using lemma~\ref{lemma:second_best}, the set of candidates for a restricted second best solution can be limited to the solutions in $\mathcal{D}_{I, O}$ that are different from the best solution in exactly one choice.
This is used to devise an algorithm \texttt{ALG-R2P} (where for any $\ske \in \skl$, $\texttt{next}_{\skl}(e_i, \ske)$ returns the next best event realization that is allowed in the current $\mathcal{D}_{I, O}$), which is shown in Alg.~\ref{alg:r2p}.
\texttt{ALG-R2P} compares the candidate solutions by (1) selecting an event realization $e_i$ in the best solution $L_q$, (2) retrieving its next best realization $e_i' \in R(\ske)$, and (3) calculating the substitution probability ratio $\rho_i = \frac{P(e_i')}{P(e_i)}$.
Then, the second best solution $L_q^2$ is the candidate solution with the highest substitution probability ratio $\rho_i$.

\noindent These auxiliary algorithms are used with the BST procedure from~\cite{hamacher_k_1985} to calculate the top-$K$ realizations of stochastically known logs.
Note that this algorithm combining event realizations under event independence is practically identical to an algorithm combining trace-level top-$K$ rankings into a log-level top-$K$ ranking under trace independence.
Thus, it could also be used to join the results of more refined trace-level top-$K$ techniques.

\section{Evaluation}
\label{sec:evaluation}

We evaluate the efficiency of the algorithm and the benefit of the top-$K$ rankings it produces in the following steps:

\begin{itemize}
    \item \textbf{EVAL1} (Efficiency): Does the algorithm improve the complexity bound, and which parameters affect the execution time? (\textbf{RQ1})
    \item \textbf{EVAL2} (Sensitivity Analysis): How are the results affected by different properties of the input data? (\textbf{RQ2, RQ3})
\end{itemize}

\subsection{\textbf{EVAL1}: Efficiency}

\noindent For EVAL1, we formally prove an upper bound for the complexity of the algorithm, starting with the auxiliary algorithms.

\begin{lemma}
\label{lemma:1p_complexity}
    \texttt{ALG-1P} is bounded by $O(\len{\skl})$
\end{lemma}

\begin{proof}
    \texttt{ALG-1P} selects the most probable realization for each stochastically known event $\ske \in \skl$, so it performs $\len{\skl}$ selections.
    As the event realizations are sorted by probability, each selection is in $O(1)$, so \texttt{ALG-1P} is bounded by $O(\len{\skl})$.
\end{proof}

\begin{lemma}
\label{lemma:r2p_complexity}
    \texttt{ALG-R2P} is bounded by $O(\len{\skl})$.
\end{lemma}

\begin{proof}
    For each stochastically known event in $\ske \in \skl$, so $\len{\skl}$ times, \texttt{ALG-R2P} (1) retrieves the next best event realization with \textbf{$\texttt{next}_L(e_i, \ske)$}, (2) calculates $\rho_i$ ($O(1)$) and (3) compares $\rho_i$ to $\rho_{\mi{max}}$ ($O(1)$).
    Consequently, the \texttt{ALG-R2P} is bounded by $O(\len{\skl} \cdot M)$ where $M$ is the complexity of $\texttt{next}_L(e_i, \ske)$.
    Because the event realizations are ranked by probability, the next best event realization is simply the next element in this events' list of realizations, making $M \in O(1)$. 
    Then, algorithm \texttt{ALG-R2P} is bounded by $O(\len{\skl})$.
\end{proof}

\begin{lemma}
\label{lemma:complexity}
    The top-$K$ realizations algorithm is bounded by $O(K \cdot \len{\skl})$.
\end{lemma}

\begin{proof}
    Hamacher and Queyranne show that the complexity of an algorithm derived from the BST procedure is in $O(C(m) + (K - 1) \cdot B(m))$ where $C(m)$ is the complexity of \texttt{ALG-1P} and $B(m)$ is the complexity of \texttt{ALG-R2P} \cite{hamacher_k_1985}.
    With lemmas \ref{lemma:1p_complexity} and \ref{lemma:r2p_complexity} follows that the full algorithm is in $O(\len{\skl} + (K - 1) \cdot \len{\skl}) = O(K \cdot \len{\skl})$.
\end{proof}

This proof formally shows the efficiency of the algorithm compared to the exponential complexity of a brute force solution.
Notably, the runtime is affected only by the number of events $\len{\skl}$ and the number of realizations to calculate $K$.

\subsection{\textbf{EVAL2}: Sensitivity Analysis}

In EVAL2, we analyze the benefit of top-$K$ rankings over top-1 interpretations.
For this, the effect of different properties of the stochastically known event logs on the distribution of probabilities in the ranking (\textbf{RQ2}) and the variability of the calculated realizations (\textbf{RQ3}) is analyzed.
The implementation of the algorithm and the code to reproduce the evaluation results are available on GitHub\footnote{\url{https://github.com/arvidle/topK_realizations}}.

\subsubsection{Simulation}

To examine the effect of different properties of the input data,
stochastically known event logs are simulated based on a modified version of the simulation procedure from \cite{DBLP:journals/eaai/FelliGMRW23}.
First, $n_{\mi{events}}$ events with random activity labels are generated.
Then, uncertainty is introduced into $\lfloor r \cdot n_{\mi{events}} \rfloor$ randomly selected events.
Each uncertain event is simulated by picking $n_{\mi{act}}$ alternative activity realizations, and assigning probabilities based on the parameter $\beta$.
The probability of the first activity $p_1$ is set to 1, and the following probabilities are generated recursively with $p_{i + 1} = p_i \cdot \beta \cdot \mi{rand}_i$ where each $\mi{rand}_i$ is a uniformly random value in the range $[0.9, 1.1]$. 
Finally, the probability values are normalized to have a sum of 1 for each event.

\subsubsection{Effect of the parameters on the ranking}

This simulation procedure is used to evaluate the effect of a varying degree of uncertainty in the input data on the resulting top-$K$ ranking.
To characterize a top-$K$ ranking, the following measures are used:
\begin{definition}[Measures]
    Let $L_1, \ldots, L_K$ be a top-$K$ ranking.
    Then, $F_K(k) = \sum_{i = 1}^{k} P(L_i)$ is the cumulative probability, and $d_{\mi{avg}} = K^{-1} \cdot \sum_{L \in \set{L_2, \ldots, L_K}}\, \len{L \setminus L_1}$ is the average number of choices different to the best realization.
\end{definition}
Additionally, the run-time of the algorithm $t$ is measured in seconds.
In the following, the effect of the different simulation parameters (properties of the input event logs) and values of $K$ on these measures is evaluated.
This is done systematically by varying each parameter separately while keeping the other parameters fixed.
All tests were executed on a computer with an Apple M2 Pro and 16 GiB memory, and the measures were averaged over 10 runs.

\begin{figure*}[t]
\centering
\subfloat{\includegraphics[width=0.24\textwidth]{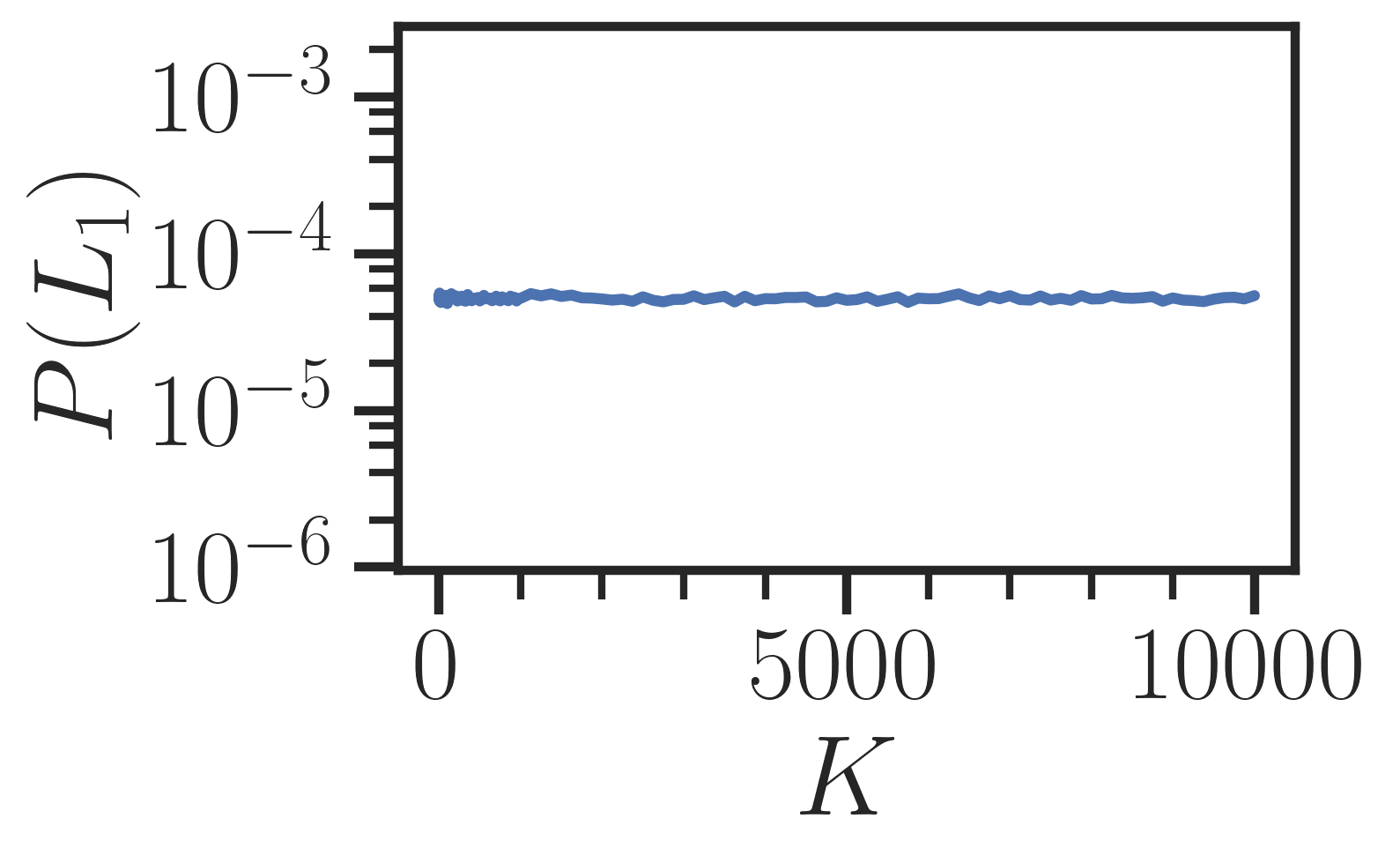}%
\label{fig:eval_k_p1}}
\subfloat{\includegraphics[width=0.24\textwidth]{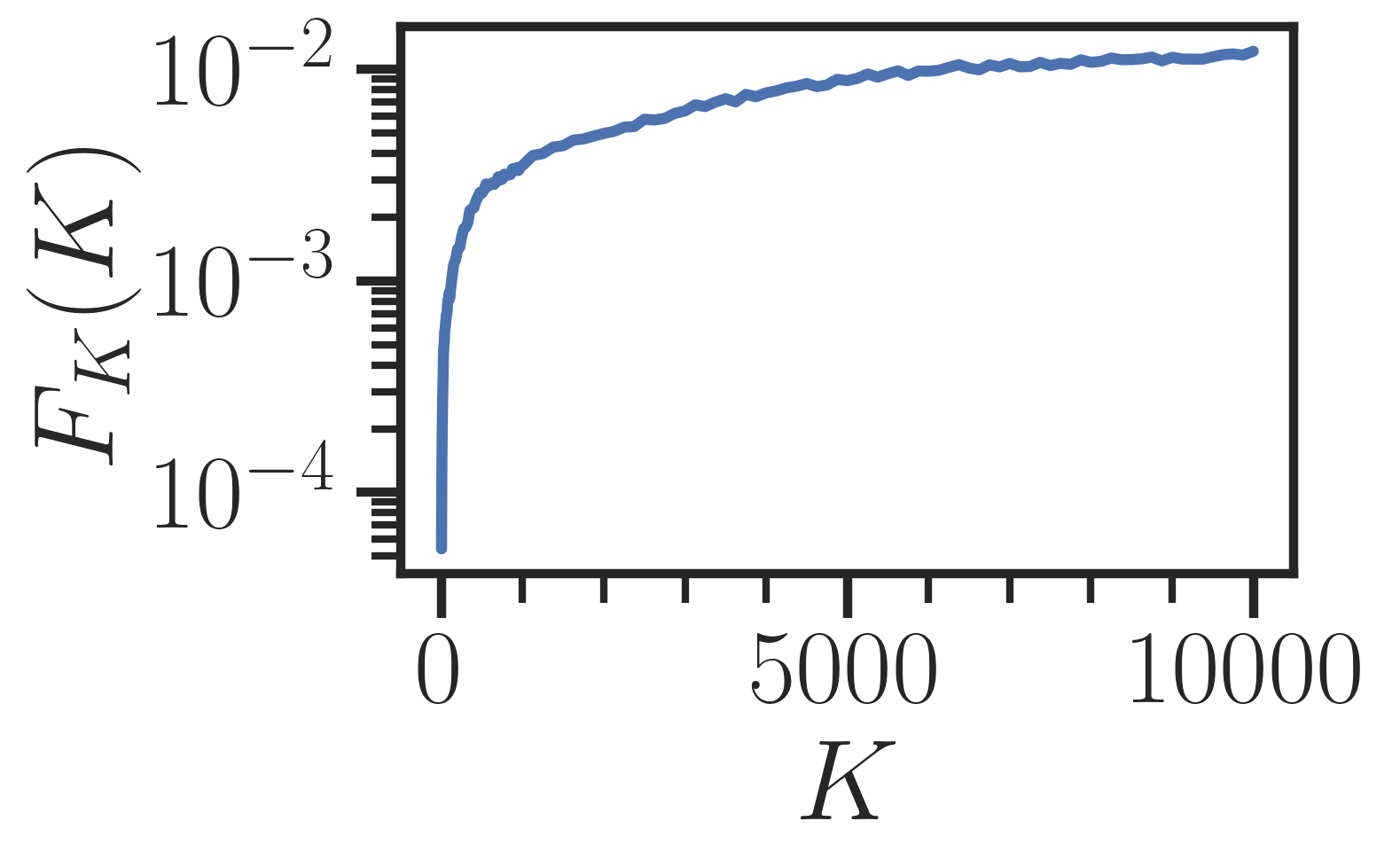}%
\label{fig:eval_ks_cp}}
\subfloat{\includegraphics[width=0.24\textwidth]{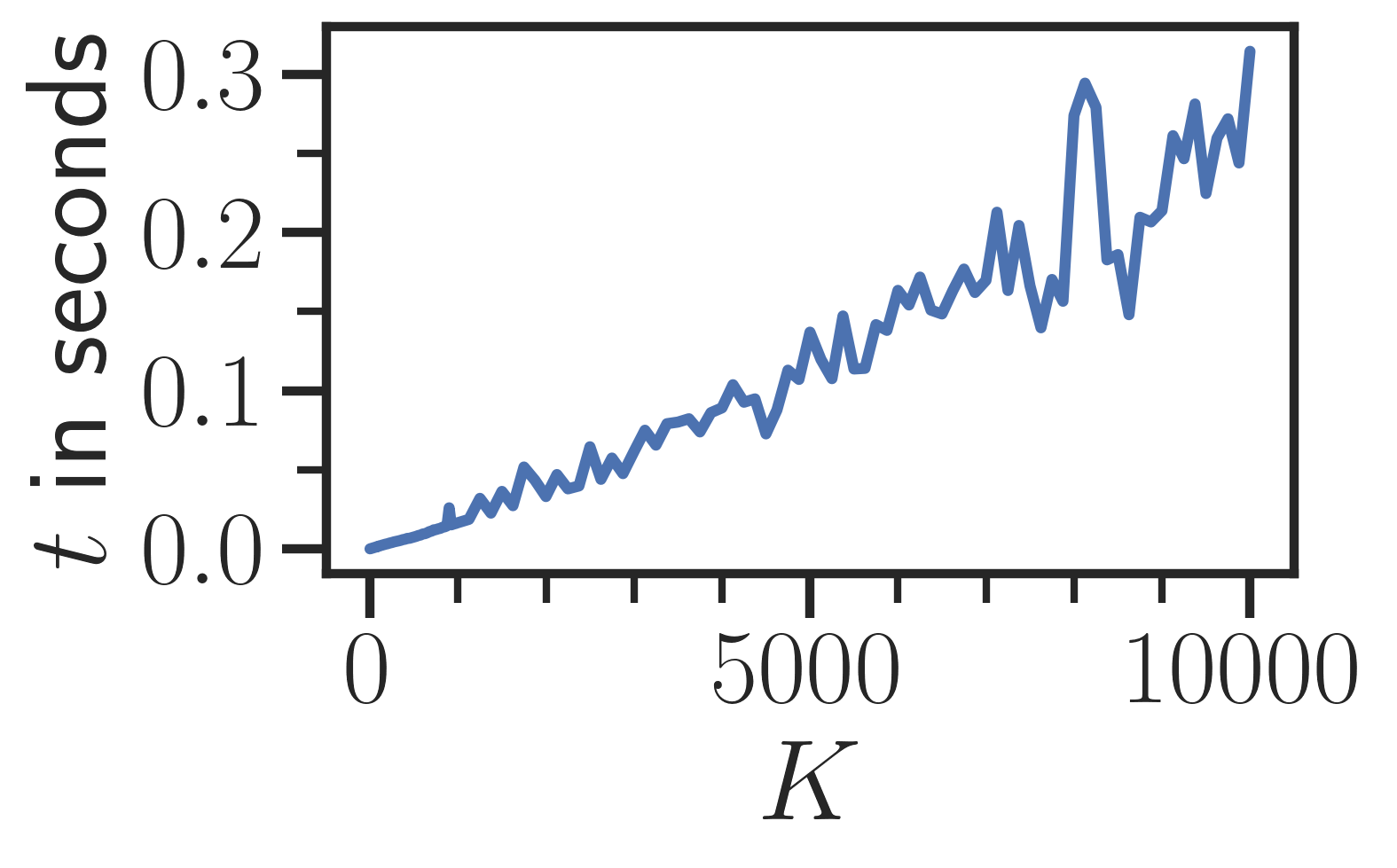}%
\label{fig:eval_k_psi}}
\subfloat{\includegraphics[width=0.24\textwidth]{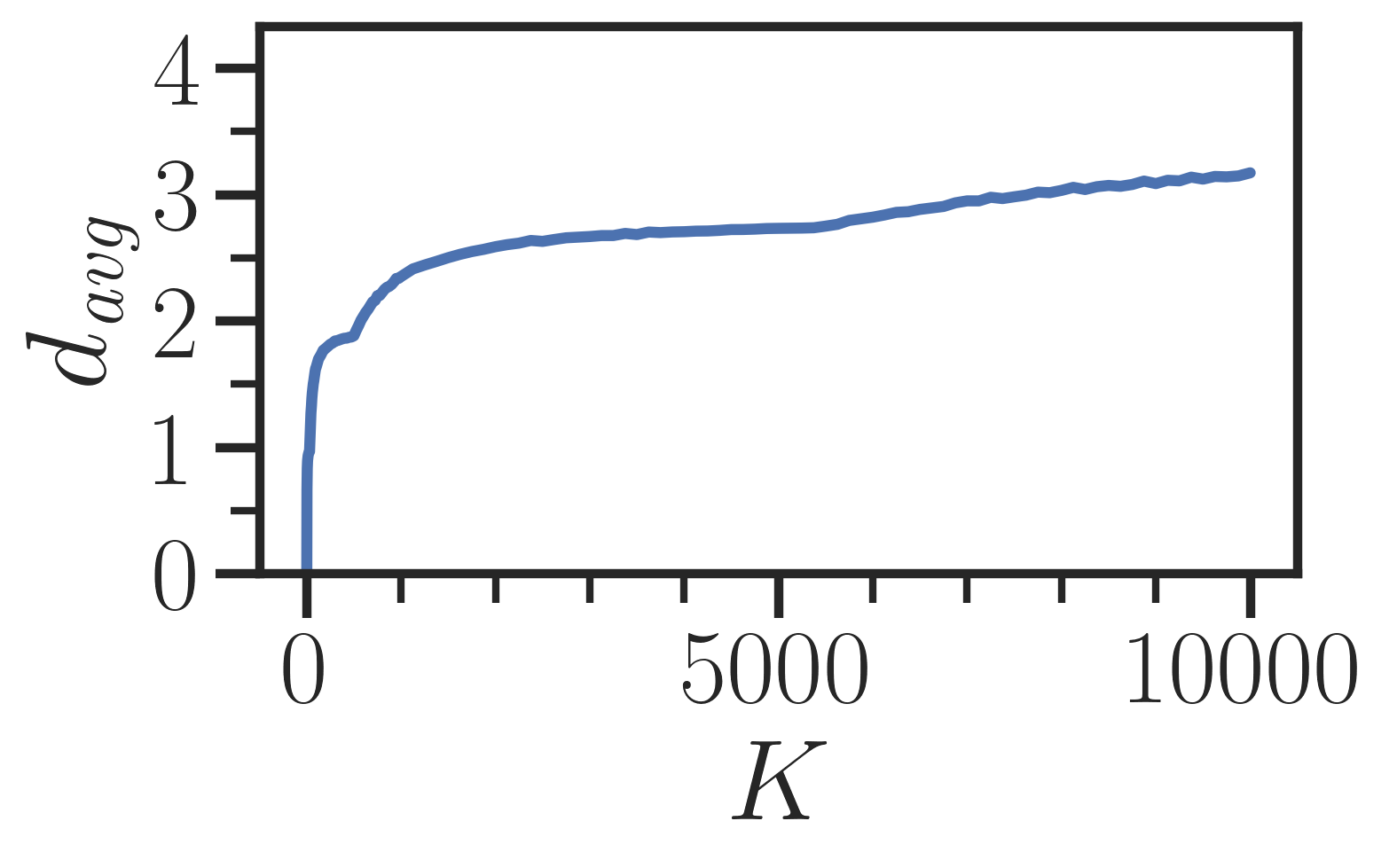}%
\label{fig:eval_k_d_avg}}
\caption{Ranking measures $P(L_1)$, $F_K(K)$, $t$ and $d_{\mi{avg}}$ for varying $K$. ($n_{\mi{events}} = 100,\, r = 0.3,\, n_{\mi{act}} = 3$ and $\beta = 0.3$; log-scaled y-axis for $P(L_1)$ and $F_K(K)$)}
\label{fig:eval_k}
\end{figure*}

First, the effect of $K$ is examined by fixing $n_{\mi{events}} = 100,\, r = 0.3,\, n_{\mi{act}} = 3$ and $\beta = 0.3$ (Fig.~\ref{fig:eval_k}).
Because $K$ is not a simulation parameter, $P(L_1)$ is constant except for small variations due to noise occurring during simulation.
The cumulative probability of the ranking $F_K(K)$ first increases sharply with $K$, but because the probabilities in the ranking are monotonically decreasing, the growth of $F_K(K)$ slows considerably with increasing $K$.
The run-time of the algorithm scales linearly with $K$.
The average difference $d_{\mi{avg}}$ also increases first, but slowly converges to a value just above 3.
At certain points, almost all realizations different to the top-1 realization in 1, 2 and 3 choices are contained in the ranking.
This causes sudden increases of the growth of $d_{\mi{avg}}$.

\begin{figure*}[t]
\centering
\subfloat{\includegraphics[width=0.24\textwidth]{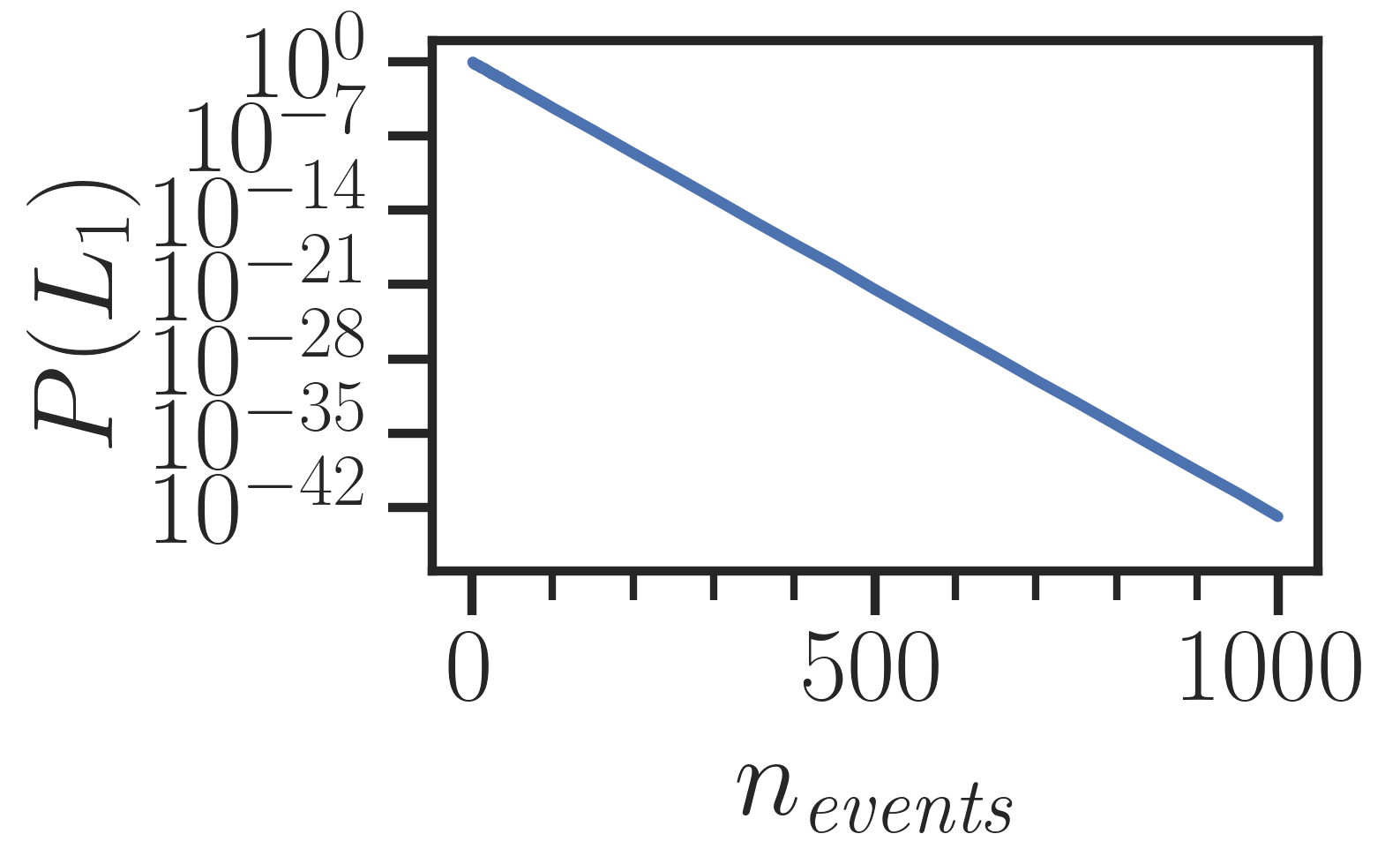}%
\label{fig:eval_n_events_p1}}
\subfloat{\includegraphics[width=0.24\textwidth]{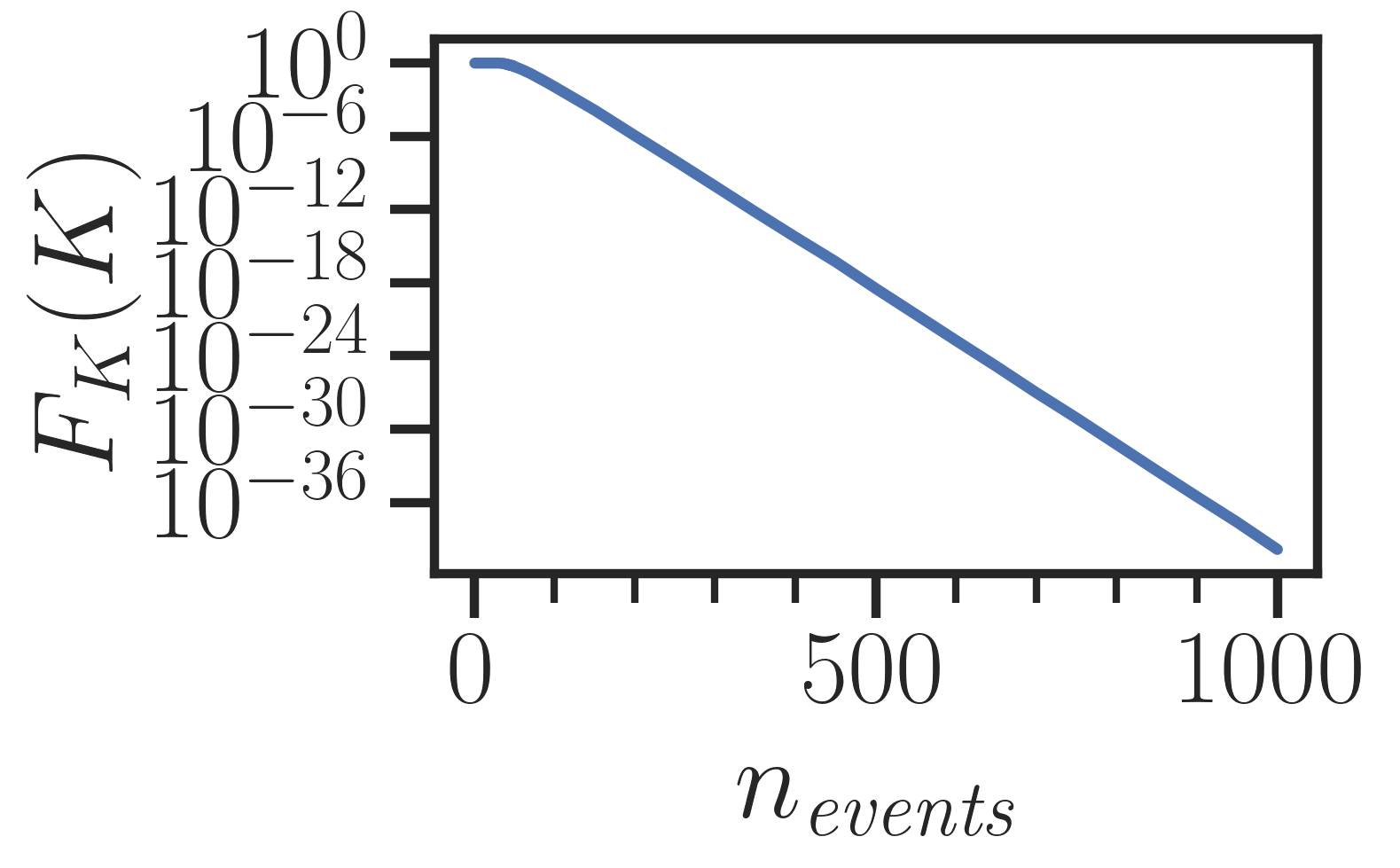}%
\label{fig:eval_n_events_cp}}
\subfloat{\includegraphics[width=0.24\textwidth]{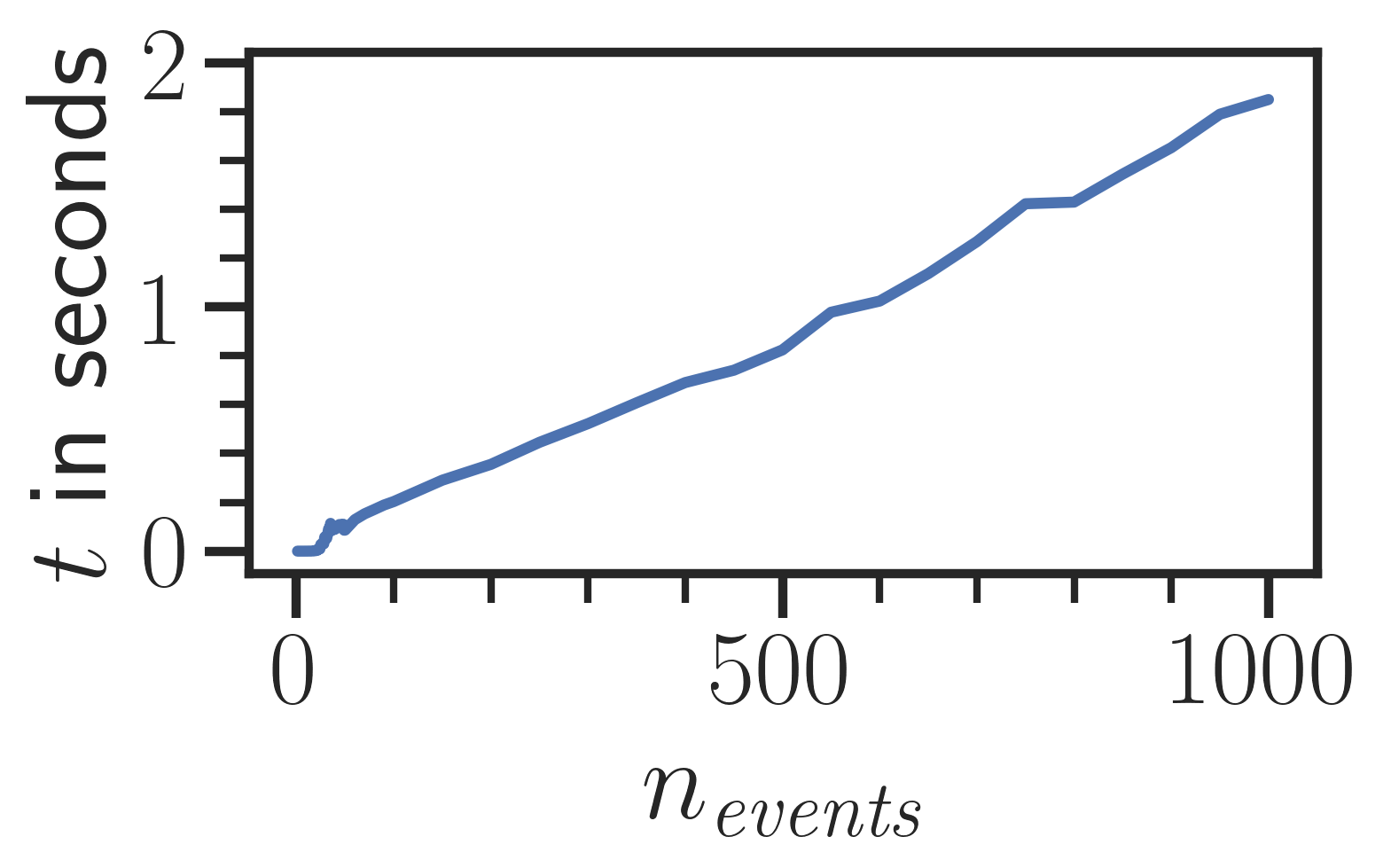}%
\label{fig:eval_n_events_runtime}}
\subfloat{\includegraphics[width=0.24\textwidth]{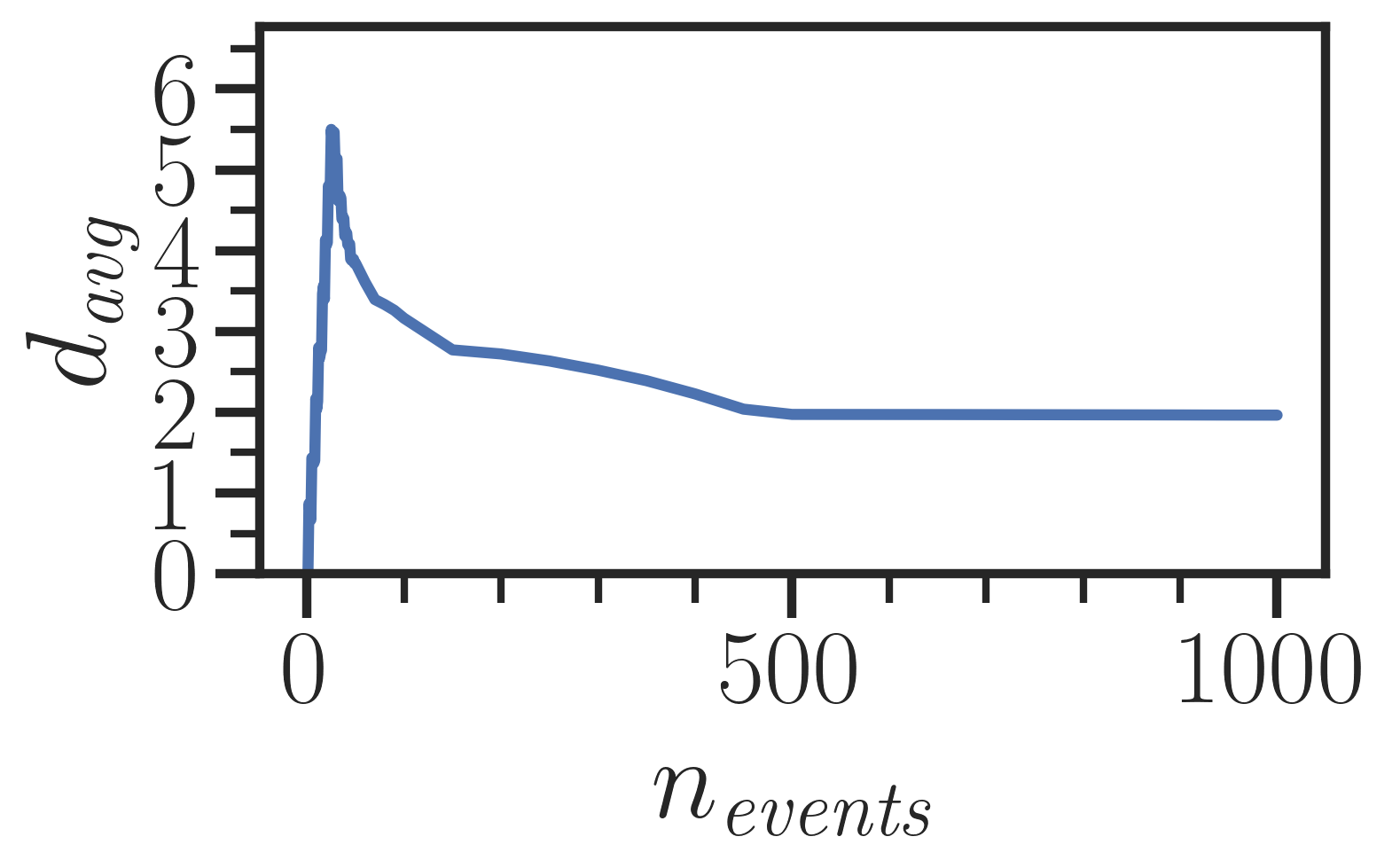}%
\label{fig:eval_n_events_d_avg}}
\caption{Ranking measures $P(L_1)$, $F_K(K)$, $t$ and $d_{\mi{avg}}$ for varying $n_{\mi{events}}$. ($r = 0.3,\, n_{\mi{act}} = 3,\, \beta = 0.3$ and $K = 10^4$; log-scaled y-axis for $P(L_1)$ and $F_K(K)$)}
\label{fig:eval_n_events}
\end{figure*}

Next, the effect of $n_{\mi{events}}$ is evaluated with $n_{\mi{act}} = 3,\, r = 0.3,\, \beta = 0.3$ and $K = 10^4$ (Fig.~\ref{fig:eval_n_events}).
Both the top-1 probability $P(L_1)$ and cumulative probability $F_K(K)$ decrease exponentially with increasing values of $n_{\mi{events}}$.
The runtime $t$ increases linearly, and the top-$10^4$ realizations of an event log with 1000 events are calculated in under 2 seconds.
The average difference $d_{\mi{avg}}$ spikes at $n_{\mi{events}} = 31$, then decreases and plateaus around a value of 2.
A low number of events forces the algorithm to choose more different realizations because the realizations which are similar to $L_1$ are exhausted.
This effect changes if more events are affected by uncertainty since the number of highly similar realizations increases rapidly with event log length.
The increase of the uncertainty threshold $r$ results in the same effects as varying $n_{\mi{events}}$.
Thus, a separate evaluation of the effect of varying $r$ is omitted.

\begin{figure*}[t]
\centering
\subfloat{\includegraphics[width=0.24\textwidth]{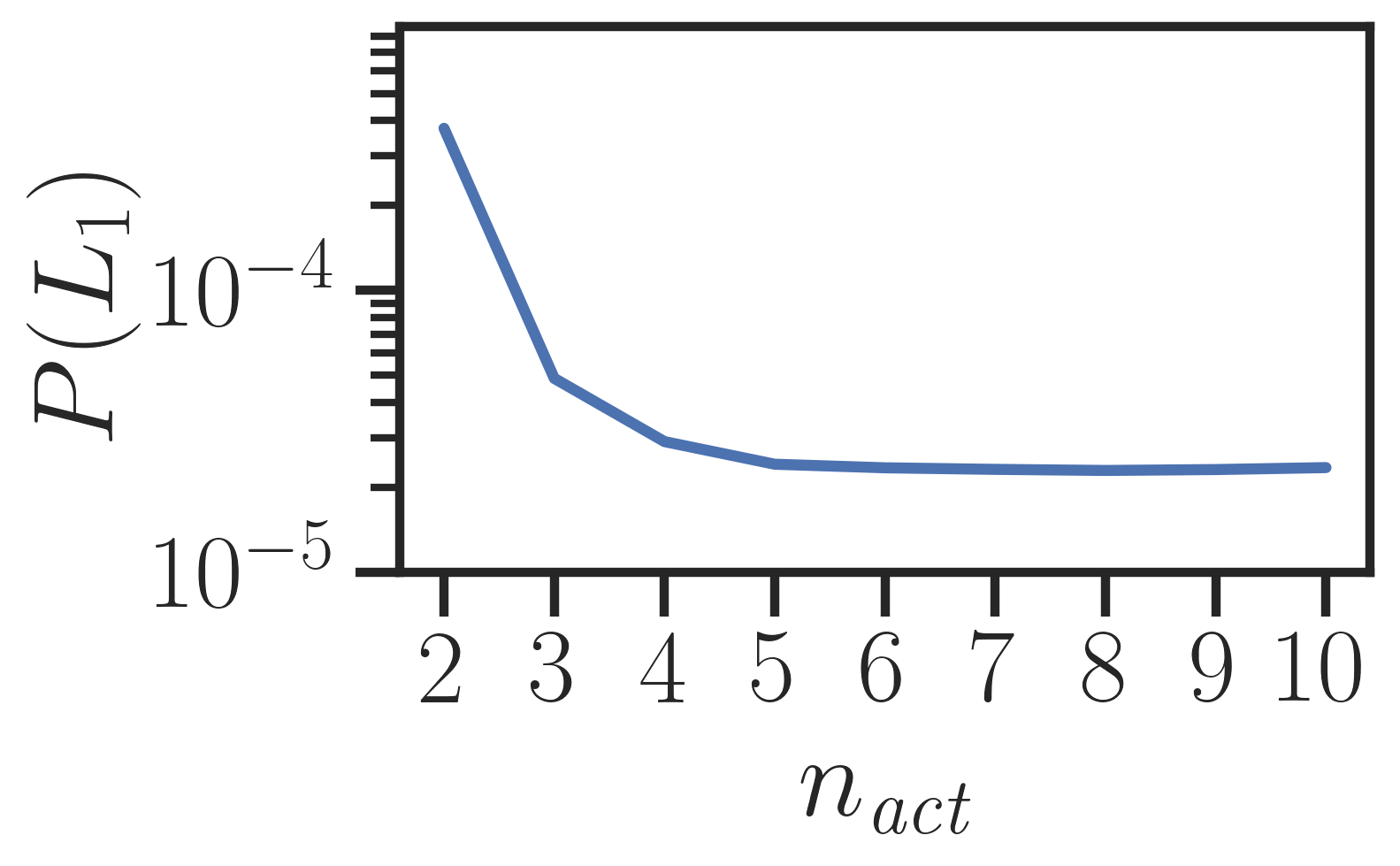}%
\label{fig:eval_n_acts_p1}}
\subfloat{\includegraphics[width=0.24\textwidth]{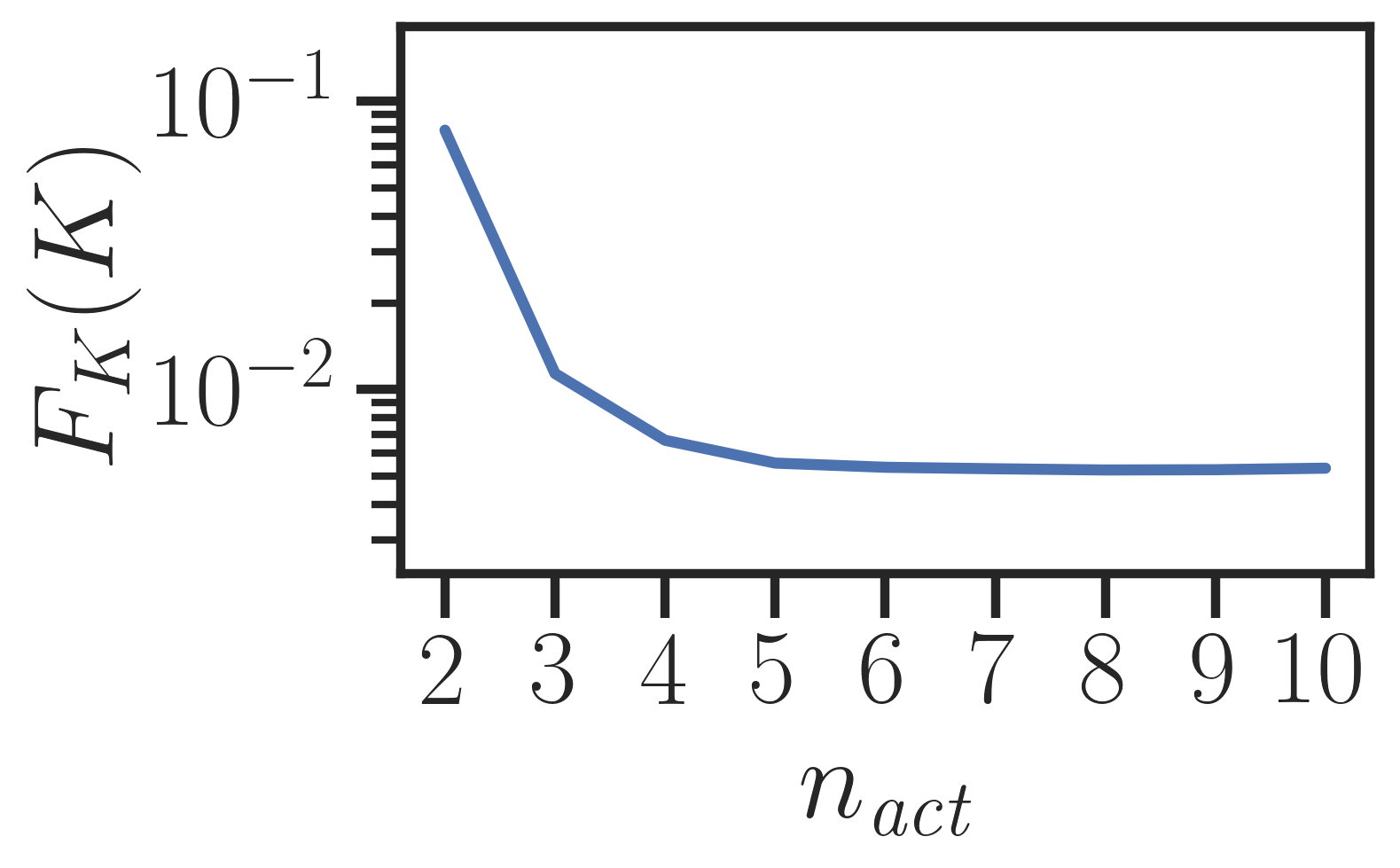}%
\label{fig:eval_n_acts_cp}}
\subfloat{\includegraphics[width=0.24\textwidth]{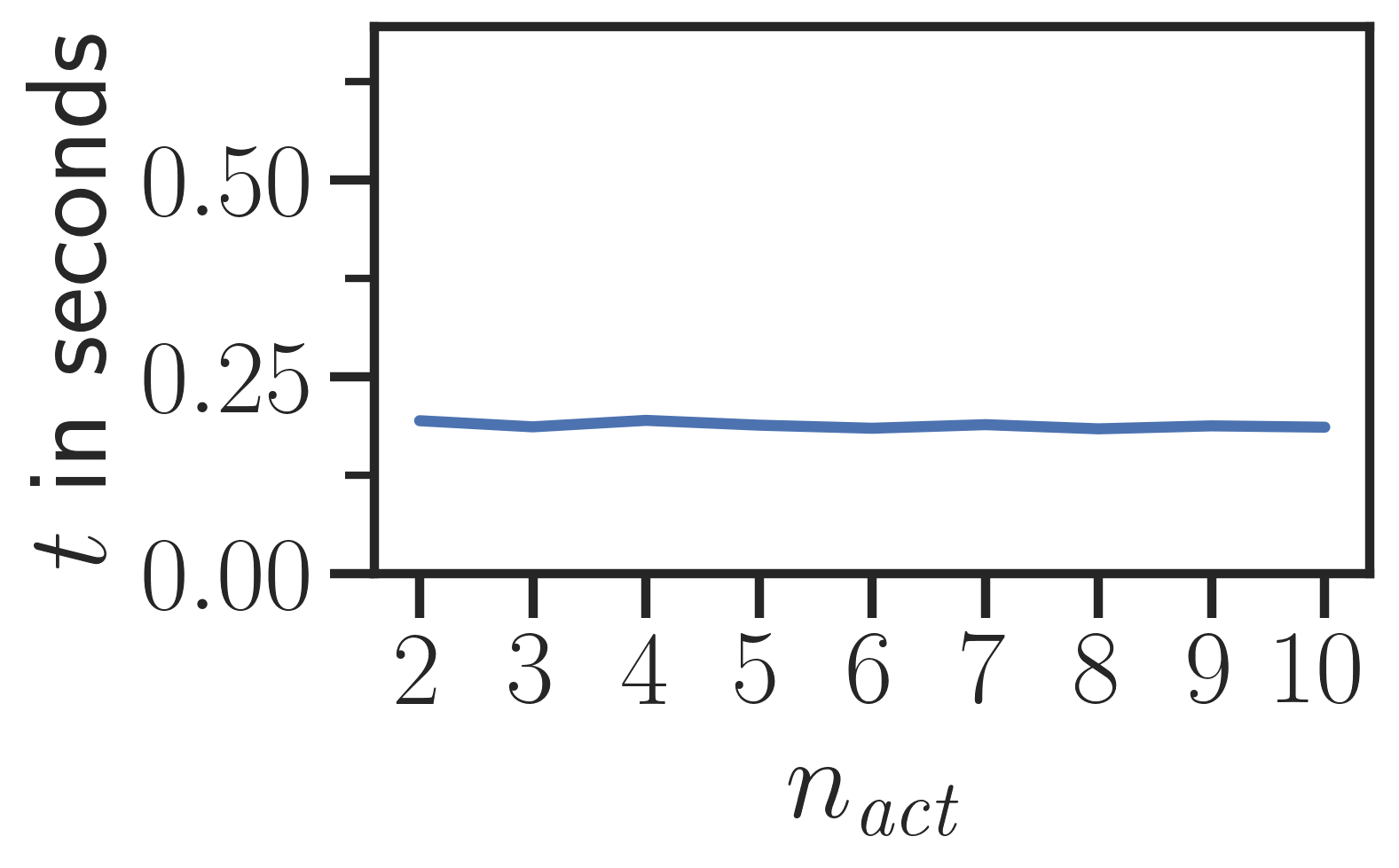}%
\label{fig:eval_n_acts_psi}}
\subfloat{\includegraphics[width=0.24\textwidth]{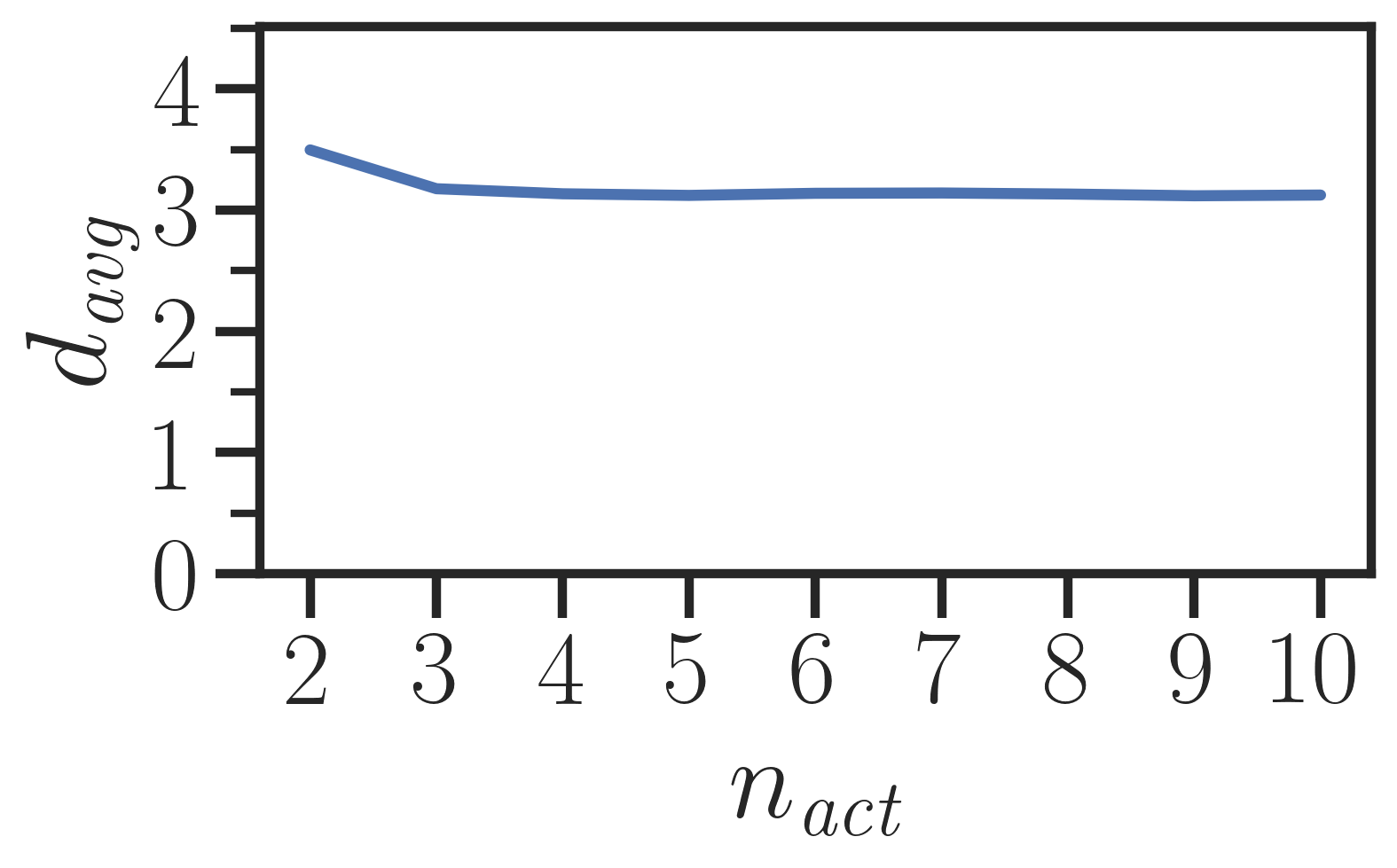}%
\label{fig:eval_n_acts_d_avg}}
\caption{Ranking measures $P(L_1)$, $F_K(K)$, $t$ and $d_{\mi{avg}}$ for varying $n_{\mi{act}}$. ($n_{\mi{events}} = 100,\, r = 0.3,\, \beta = 0.3$ and $K = 10^4$; log-scaled y-axis for $P(L_1)$ and $F_K(K)$)}
\label{fig:eval_n_acts}
\end{figure*}

Then, the number of alternatives for each stochastically known event $n_{\mi{act}}$ is evaluated for $n_{\mi{events}} = 100,\, r = 0.3,\, \beta = 0.3$ and $K = 10^4$.
Both $P(L_1)$ and $F_K(K)$ decrease with increasing values of $n_{\mi{act}}$.
Because the probabilities inside each event must sum up to 1, with higher number of alternatives, the probability is distributed among more alternatives (and the number of possible realizations increases significantly), making all outcomes less probable.
The run-time of the algorithm is constant with respect to $n_{\mi{act}}$.
For lower values of $n_{\mi{act}}$, $d_{\mi{avg}}$ slightly elevates because it is more likely that the alternatives of an event have all been picked and thus another event needs to be changed to generate new candidates for the next best solution.

\begin{figure*}[t]
\centering
\subfloat{\includegraphics[width=0.24\textwidth]{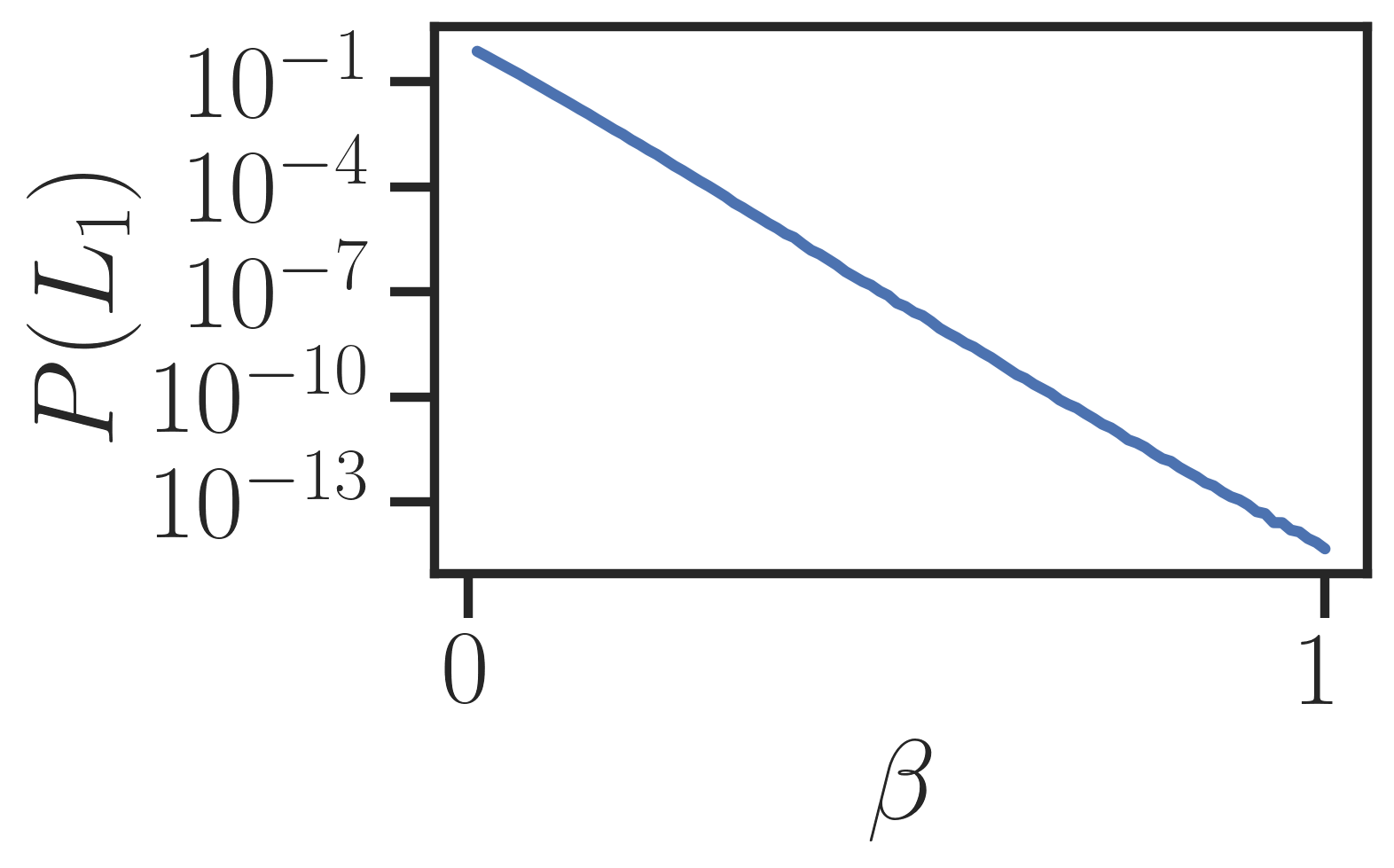}%
\label{fig:eval_beta_p1}}
\subfloat{\includegraphics[width=0.24\textwidth]{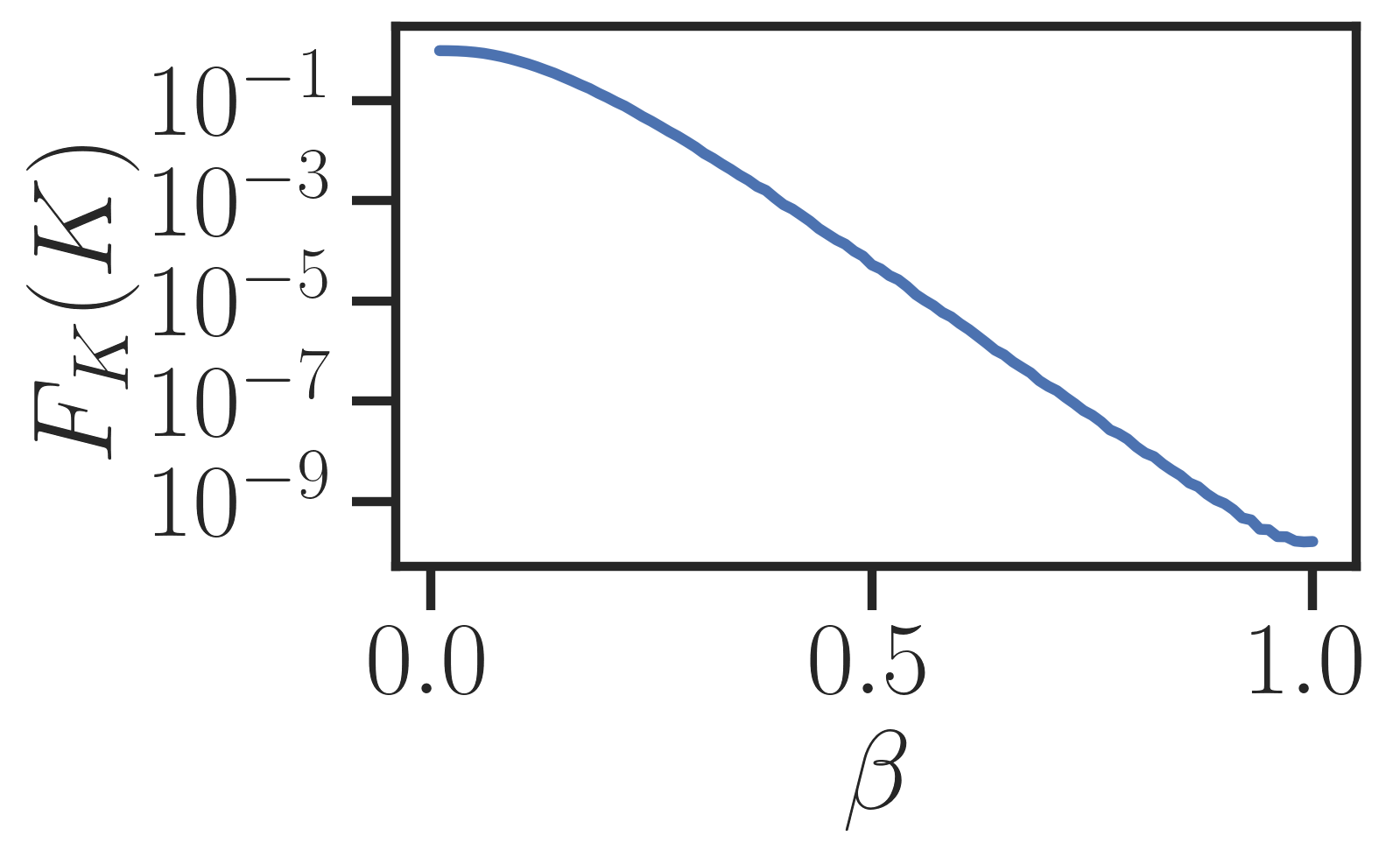}%
\label{fig:eval_beta_cp}}
\subfloat{\includegraphics[width=0.24\textwidth]{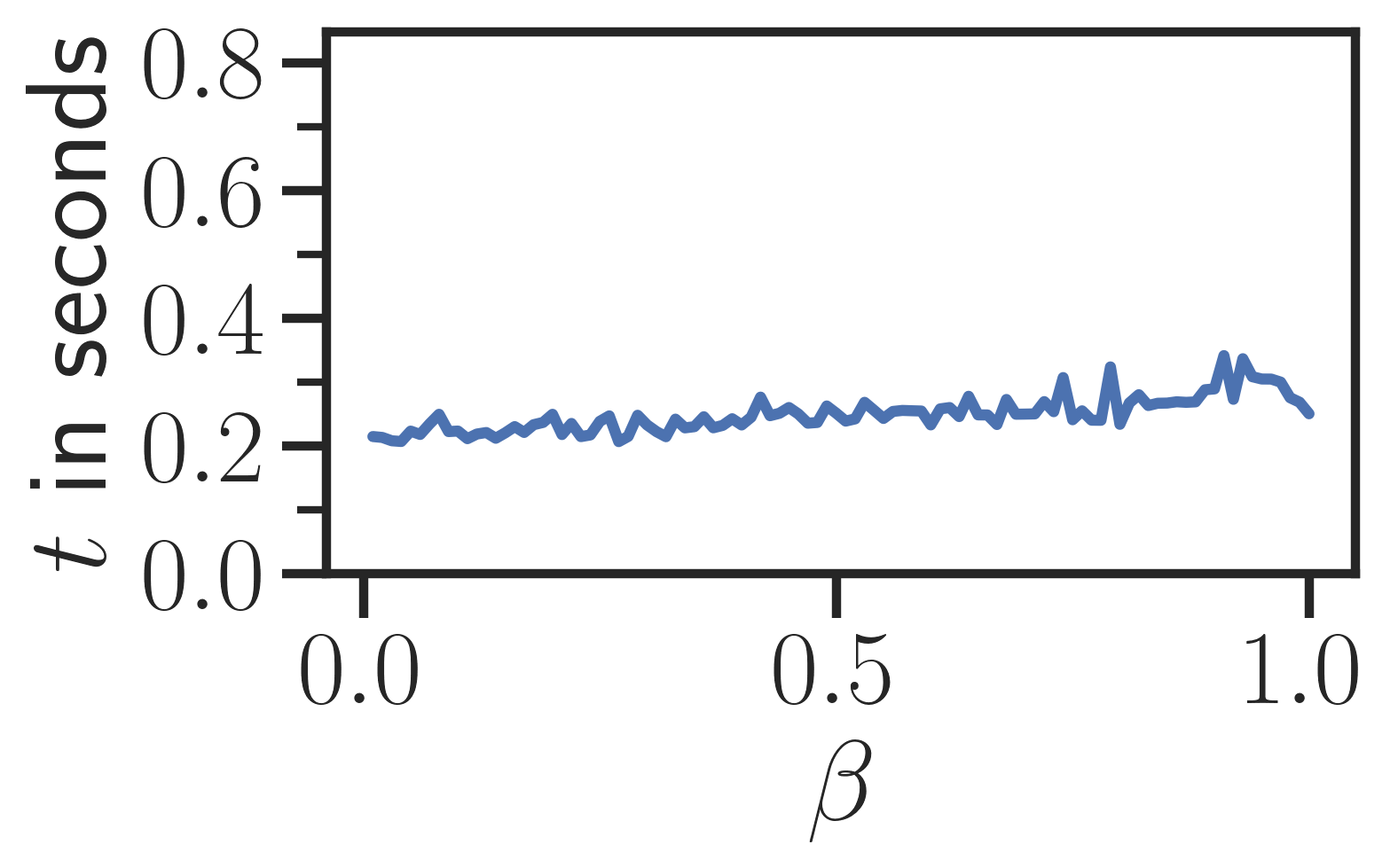}%
\label{fig:eval_beta_runtime}}
\subfloat{\includegraphics[width=0.24\textwidth]{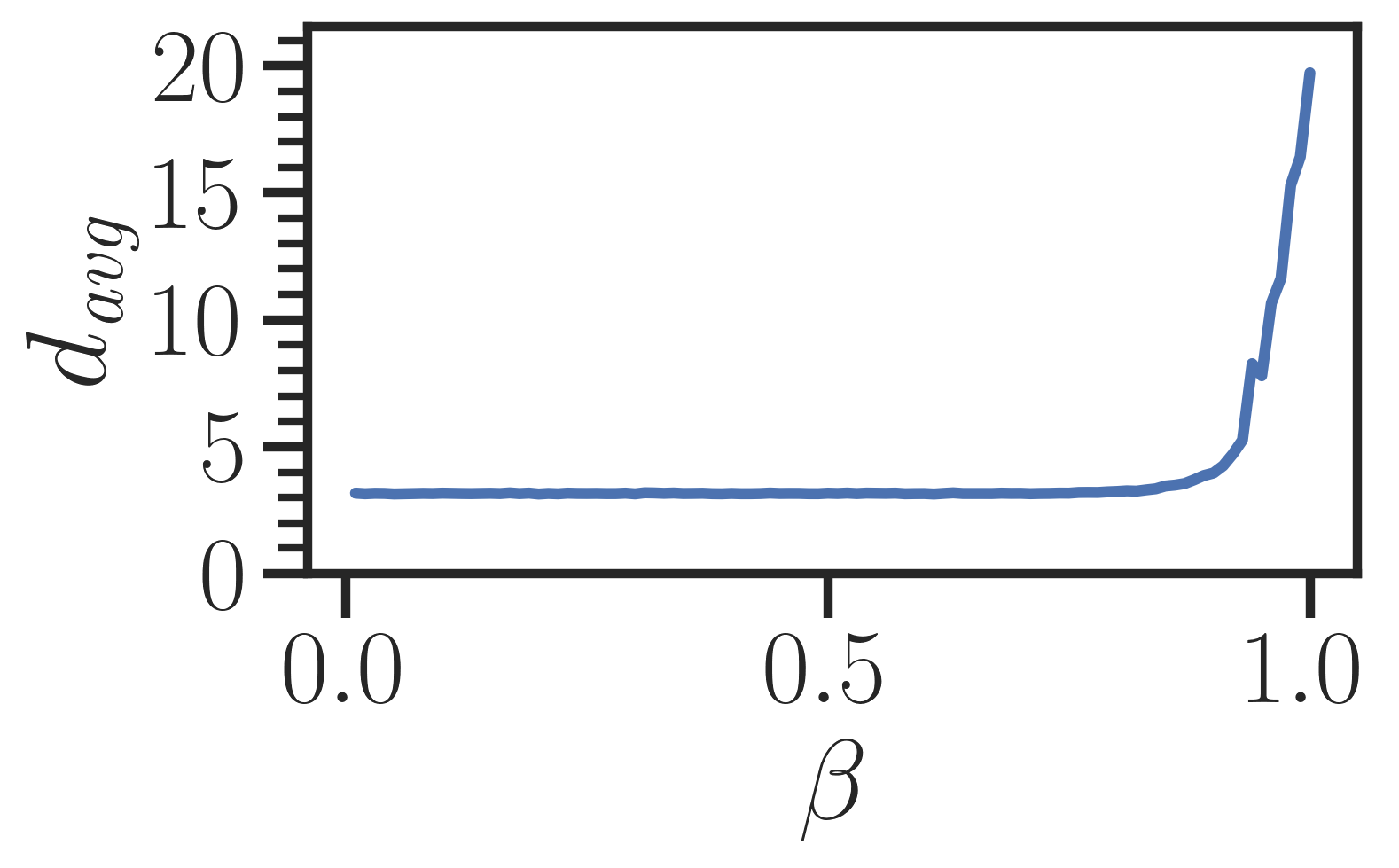}%
\label{fig:eval_beta_d_avg}}
\caption{Ranking measures $P(L_1)$, $F_K(K)$, $t$ and $d_{\mi{avg}}$ for varying $\beta$. ($n_{\mi{events}} = 100,\, r = 0.3,\, n_{\mi{act}} = 3$ and $K = 10^4$; log-scaled y-axis for $P(L_1)$ and $F_K(K)$)}
\label{fig:eval_beta}
\end{figure*}

Finally, the effect of $\beta$ is evaluated with  $n_{\mi{events}} = 100,\, r = 0.3,\, n_{\mi{act}} = 3$ and $K = 10^4$.
$P(L_1)$ and $F_K(K)$ both decrease for increasing values of $\beta$, i.e., less skewed event probability distributions.
The run-time of the algorithm is largely constant with respect to $\beta$.
The average difference $d_{\mi{avg}}$ is unaffected up to $\beta = 0.8$ and then increases sharply.
Because the probability values of the event alternatives are very similar at this point, there are many different realizations having nearly equal probabilities.

\subsection{Discussion}

Overall, the sensitivity analysis confirms that the run-time of the algorithm scales linearly with both $n_{\mi{events}}$ and $K$.
Even for larger optimization problems ($n_{\mi{events}} = 1000$, $K = 10^4$), the algorithm calculates a top-$K$ ranking in around 2 seconds, confirming its efficiency.
To further evaluate the efficiency of the algorithm, it was applied to two uncertain event logs presented in~\cite{engelberg_uncertainty-aware_2023}.
The calculation of the top-$10^4$ realizations took $0.366$ and $0.505$ seconds respectively ($33.5$ and $43.16$ seconds for top-$10^6$).
For comparison, we implemented a baseline approach as described in section~\ref{algorithm}.
While this baseline algorithm performs reasonably well for smaller logs (e.g., $5.7s$ for $n_{\mi{events}} = 40$, $n_{\mi{act}} = 3$, $K = 10^4$), its run-time and memory requirements scale exponentially, making its application infeasible for larger input logs.
In fact, the baseline algorithm was unable to calculate rankings for the logs from~\cite{engelberg_uncertainty-aware_2023} because of memory limitations.
Generally, additional realizations can be produced significantly faster by our algorithm than they can be processed, e.g., using process discovery algorithms.
Thus, the algorithm shows how top-$K$ rankings can be calculated efficiently, giving an answer to \textbf{RQ1}.

Regarding \textbf{RQ2}, we evaluated the top-1 probability $P(L_1)$ and the cumulative probability $F_K(K)$.
Both decrease rapidly with increasing uncertainty and size of the input event log.
Consequently, a direct application of the top-$K$ algorithm to cover a representative set of the realizations by probability is only sensible for small log sizes (see Fig.~\ref{fig:eval_n_events_cp}).
For instance, a top-10,000 ranking of event logs simulated with $n_{\mi{events}} = 50,\, n_{\mi{act}} = 3,\, \beta = 0.3$ and $r = 0.3$ on average covers about 53\% of the realizations by probability while containing only about $0.07\%$ of the $1.43 \cdot 10^7$ possible realizations.
However, even for larger and more uncertain event logs, top-$K$ interpretations are beneficial over the most probable realization, with $F_K(K)$ being consistently larger than $P(L_1)$ by around 3 orders of magnitude for $K = 10^4$.
From this, it can be concluded that the challenge of exponentially decreasing probabilities is not specific to top-$K$ interpretations, but rather a general challenge when handling uncertain event data.
The diminishing returns for $F_K(K)$ with increasing $K$ constitute another challenge -- an increase in uncertainty or size of the input event log cannot be compensated with a proportional increase of $K$.

For \textbf{RQ3}, we observe that low values result for $d_{\mi{avg}}$ ($\leq 4$ for sensible parameters).
Because the realizations of even a small event log (e.g., 100 events) differ only slightly, the information gain of a top-$K$ ranking in terms of the variability of the realizations appears limited for sensible values of $K$.

In summary, the benefit of top-$K$ realizations is most pronounced for smaller event logs and logs with low degrees of uncertainty.
However, even for larger event logs, top-$K$ rankings consistently provide a benefit over top-$1$ interpretations.

\section{Conclusion}
\label{conclusion}

In this paper, we presented an algorithm to compute the top-$K$ most probable realizations of stochastically known event logs. 
We also evaluated the benefit of top-$K$ rankings against top-$1$ interpretations of stochastically known event logs.
We formally proved that our algorithm operates within a computational complexity of $O(K \cdot \len{\skl})$ (see EVAL1), which builds a foundation for future research on uncertainty-aware process mining techniques.
We also showed that top-$K$ interpretations of stochastically known event logs provide a benefit over using the single most probable realization, especially for smaller logs or isolated cases of larger event logs with moderate incidences of uncertainty (see EVAL2). 

To allow wider application of top-K rankings in process mining on uncertain event data, the incorporation of more complex event and trace dependencies into our algorithm might constitute an avenue for future research. 
Secondly, in order to improve the variability of the realizations, the algorithm could be extended with techniques to diversify its outputs~\cite{qin_diversifying_2012}.

%
\bibliographystyle{splncs04}
\bibliography{bibliography}
\end{document}